%% file: main.tex
\documentclass[11pt]{article}
\input{defs}

\usepackage{fullpage}

\title{A Near-Optimal Offline Algorithm for Dynamic All-Pairs Shortest Paths in Planar Digraphs} 

\author{Debarati Das\thanks{University of Copenhagen, Denmark, das@di.ku.dk.} \and Maximilian Probst Gutenberg\thanks{ETH Zurich, Switzerland, maximilian.probst@outlook.com. Supported by Rasmus Kyng's Start-up Grant at ETH. Work done while at BARC, supported by Thorup's Investigator Grant No. 16582.} \and Christian Wulff-Nilsen\thanks{University of Copenhagen, Denmark, koolooz@di.ku.dk. The author is supported by the Starting Grant 7027-00050B from the Independent Research Fund Denmark under the Sapere Aude research career programme.}}

\newcommand{\DDG}[1]{\ensuremath{\mbox{DDG}({#1})}}

\newtheorem{research question}{Research Question}

\def\poly{\operatorname{poly}}

\begin{document}
\date{}
\maketitle
\pagenumbering{gobble}

\begin{abstract}
In the planar, dynamic All-Pairs Shortest Paths (APSP) problem, a planar, weighted digraph $G$ undergoes a sequence of edge weight updates and the goal is to maintain a data structure on $G$, that can quickly answer distance queries between any two vertices $x,y \in V(G)$. 

The currently best algorithms [\hyperlink{cite.fakcharoenphol2001planar}{FOCS'01}, \hyperlink{cite.klein2005multiple}{SODA'05}] for this problem require $\tilde{O}(n^{2/3})$ worst-case update and query time, while conditional lower bounds [\hyperlink{cite.abboud2016popular}{FOCS'16}] show that either update or query time 
$\tilde{\Omega}(\sqrt{n})$
is needed\footnote{We use $\tilde{O}$ and $\tilde{\Omega}$-notations to hide poly-logarithmic factors in $n$. 
}.

In this article, we present the first algorithm with near-optimal $\tilde{O}(\sqrt{n})$ worst-case update and query time for the offline setting, where the update sequence is given initially. This result is obtained by giving the first offline dynamic algorithm for maintaining dense distance graphs (DDGs) faster than recomputing from scratch after each update.

Further, we also present an \emph{online} algorithm for the incremental APSP problem with $\tilde{O}(\sqrt{n})$ worst-case update/ query time. This allows us to reduce the online dynamic APSP problem to the online decremental APSP problem, which constitutes partial progress even for the online version of this notorious problem.
\end{abstract}

\newpage
\pagenumbering{arabic}

\section{Introduction}

In the planar, dynamic All-Pairs Shortest Paths (APSP) problem, one is given a planar, directed and weighted graph $G = (V,E,w)$ undergoing edge updates, i.e. edge insertions and deletions and the goal is to maintain a data structure that processes edge updates and can return for any pair of vertices $x,y \in V$, the distance $d_G(x,y)$ from $x$ to $y$ in the current version of $G$. In this work we use a somewhat restricted setting where the embedding is preserved throughout. Thus the allowed updates are edge weight increases and decreases.

This problem is maybe the most notorious open problem in planar dynamic graph algorithms: while a data structure exists \cite{fakcharoenphol2001planar, klein2005multiple} that achieves worst-case update and query time $\Tilde{O}(n^{2/3})$ and a conditional lower bound \cite{abboud2016popular} indicates that 
$\tilde{\Omega}(\sqrt{n})$
update and query time is essentially necessary, progress to close this gap has proven elusive. 

\begin{research question}
Can we close the gap between $\tilde{O}(n^{2/3})$ and 
$\tilde{\Omega}(\sqrt{n})$
update and query time for planar, dynamic APSP?
\end{research question}

Currently all approaches to achieve fast update and query times boil down to rather straight-forward extensions of static algorithms combined with the simple fact that there exists an $r$-division for any planar graph. Our approach involves a recursive decomposition into $r$-divisions.

\begin{definition}[r-division]\label{def:RDiv} Given planar embedded graph $G$ and $r \geq 1$, an $r$-division is a set of $O(n/r)$ edge-induced subgraphs $G_1 = G[E_1], G_2 = G[E_2], \dots, G_k = G[E_k]$, called \emph{pieces}, each of vertex size at most $r$ where $E_1,E_2, \dots, E_k$ partition the edge set $E$, such that each piece $G_i$ shares at most $O(\sqrt{r})$ vertices $\partial G_i$ with any of the other pieces. Further all vertices $\partial G_i$ are on a constant number of holes of $G_i$ (where a \emph{hole} of $G_i$ is a face in the induced graph $G_i$ that is not a face of $G$).
\end{definition}
For convenience, we will assume the vertices of the $O(1)$ holes of a piece $P$ are exactly the boundary vertices of $P$. This can always be ensured by adding, for each hole $h$, infinite-weight dummy edges between boundary vertices of $h$ that are consecutive in the cyclic ordering of the vertices of $h$.
\begin{definition}[Capped Recursive decomposition]\label{def:cappedRDecomp}
Given graph $G$ and $r \geq 1$, an \emph{$r$-capped recursive decomposition} of $G$ is a forest $\mathcal{R}$ with $\lfloor \lg r \rfloor$ levels where each node is associated with a subgraph of $G$. In particular, each child node is associated with a subgraph of the graph associated with its parent node. Further, the nodes in $\mathcal{R}$ that are at distance $i$ from the root, form an $r / 2^i$-division of $G$ (and an $r / 2^i$-division of their parent's graph) and leaf nodes consist of a single edge only. We often do not distinguish between a node and its associated subgraph of $G$.
\end{definition}

It is well-known that such an $r$-capped recursive decomposition can be found for any $r$ and planar $G$ in linear time \cite{klein2013structured}. The approach of \cite{fakcharoenphol2001planar, klein2005multiple} is to initially construct an $n^{2/3}$-capped recursive decomposition $\mathcal{R}$, where for each subgraph $H$ in $\mathcal{R}$, it computes the distances $\partial H \times \partial H$ which are stored as a complete graph over $\partial H$ (so-called \emph{dense distance graph (DDG)} of $H$), and finally exploit a clever modification of Dijkstra's algorithm (see the Theorem below) to implement queries in $\tilde{O}(n^{2/3})$ time by leveraging the DDGs.

\begin{theorem}[FR-Dijkstra, see \cite{fakcharoenphol2001planar,kaplan2012submatrix}]\label{thm:FRDijkstra}
Given a graph $G$, and $r$-capped recursive decomposition $\mathcal{R}$ and dense distance graphs for every graph in  $\mathcal{R}$. Then, there exists an algorithm that for any $x,y \in V$, computes the distance $d_G(x,y)$ in time $\tilde{O}(n/\sqrt{r})$. 
\end{theorem}

Both \cite{fakcharoenphol2001planar, klein2005multiple} give a (static) algorithm to compute for a given graph $H$ that is associated with a root in $\mathcal{R}$, the DDG for every $H'$ in $\mathcal{R}$ with $H' \subseteq H$ in total time $\tilde{O}(r)$; since there are $O(n/r)$ pieces $H$, the total preprocessing time is thus $\tilde{O}(n)$. Further, after an edge weight change to $G$, the algorithm only has to locate the affected root graph $H$ containing the edge, and one can simply recompute the dense distances graphs in the tree rooted at $H$ from scratch, again in $\tilde{O}(r)$ time.

Thus, previous APSP algorithms maintained the DDGs of the pieces essentially by applying a static algorithm at every update step and obtained a dynamic APSP algorithm rather using the $r$-divisions technique in planar graphs. We ask whether there is a way of maintaining DDGs that does not require recomputation from scratch but rather maintains them dynamically. More concretely, we pose the following question.

\begin{research question}
Can a DDG be maintained dynamically with sublinear update time (and small query time)?
\end{research question}

We note that so far, the question about whether there exists an efficient encoding of DDG that can be maintained during updates with only a small number of bits changed over each update is not known either. In particular, for any DDG over an $n$-vertex graph $G$, there is no known encoding that does not have to change $\Omega(n)$ bits on average for each update.

\section{Contributions}

In this article, we present the first partial progress on this long-standing problem. We give the first algorithm that achieves the desired update/query time in the offline setting\footnote{Here, we only assume that the update seqeunce is offline, i.e. is initially presented the data structure, while queries can be online.}. Moreover, if the graph is \emph{incremental}, i.e. the update sequence only consists of edge weight decreases, the data structure works in the online setting. Due to the lower bound of Abboud and Dahlgaard \cite{abboud2016popular} that extends for both  offline and online incremental setting, our bounds are tight up to polylogarithmic factor.

\begin{theorem}\label{thm:mainAPSP}[Near-Optimal Dynamic APSP]
There exists a data structure for the planar, dynamic APSP problem in the offline setting, i.e. where the update sequence is accessible from the beginning of the algorithm, that has worst-case update time and query time $\tilde{O}(\sqrt{n})$ and preprocessing time $\tilde{O}(n)$. The data structure is deterministic.

If the graph is incremental then our data structure works in the online setting.
\end{theorem}

\begin{remark}
More generally, for any $\delta \in [0,\frac{1}{2}]$, we obtain an offline dynamic APSP data structure with worst-case update time $\tilde{O}(n^{\delta})$ and query time $\tilde{O}(n^{1-\delta})$. By the lower bound by Abboud and Dahlgaard \cite{abboud2016popular} which extends to the offline setting, this is tight for all choices of $\delta$. Further, we improve for every constant $\delta \in (0, \frac{1}{2}]$ on the previous state-of-the-art.
\end{remark}

Even more interestingly, we achieve our result by designing a new algorithm that efficiently maintains DDGs. We remind the reader that so far, even maintaining an encoding undergoing a sublinear number of changes in $n$ was not known.

\begin{theorem}\label{thm:mainDDG}
Given a planar embedded graph $G$ along with an $r$-capped recursive division $\mathcal{R}$ of $G$, there exists a data structure that maintains the DDG of every node $H$ in $\mathcal{R}$ in the offline setting, such that each update is processed in $\tilde{O}(\sqrt{r})$ worst-case update time, and the weight of each edge $(x,y)$ with $x,y \in \partial H$ can be returned in query time $O(\log n)$. Moreover the data structure is deterministic and has a preprocessing time $\tilde{O}(n)$. 

If the graph is incremental, then our data structure also works in the online setting.
\end{theorem}

It is not hard to see from the discussion above how this result can be combined with Theorem \ref{thm:FRDijkstra} to derive Theorem \ref{thm:mainAPSP} as a corollary. Thus, the following sections are concerned with proving Theorem \ref{thm:mainDDG}.\\

Our result can be applied to obtain an improvement to the following variant of dynamic max flow/min st-cut:
\begin{corollary}\label{cor:MaxFlow}
Given an $n$-vertex undirected planar flow network $G$ with edge capacities, a source $s$, and a sink $t$, there is a deterministic data structure that can maintain the value of a max flow from $s$ to $t$ (equivalently, the value of a min $st$-cut) under either an offline sequence of edge weight changes or an online sequence of edge weight decreases. Each update and each query for the current max flow value can be handled in $\tilde O(\sqrt n)$ worst-case time.
\end{corollary}
The previous best update/query time bound was $\tilde O(n^{2/3})$ by Italiano et al.\cite{ItalianoNSW11}. Corollary~\ref{cor:MaxFlow} follows easily by replacing the trivial dynamic DDG algorithm in~\cite{ItalianoNSW11} with our data structure; the proof can be found in Appendix~\ref{sec:DynMaxFlow}.

\paragraph{Conceptual Contribution.} We believe that our techniques and conceptual ideas give first indications towards designing such algorithms in the online setting: 
\begin{itemize}
    \item Our main conceptual change is a shift in perspective to view DDGs as a data structure and allow for small (but non-constant) query times. This is in stark contrast to previous literature where the DDGs where mostly seen as complete graphs that where explicitly constructed.
    \item Further, our online incremental DDG data structure can be used to reduce the online dynamic APSP problem to the online decremental DDG setting. More precisely, if one can give an online algorithm $\mathcal{A}$ that maintains a DDG on a graph subject to edge weight increases with worst-case update time $\tilde{O}(n^{1-\epsilon})$ for any $\epsilon > 0$, query time $\tilde{O}(1)$, and preprocessing $\tilde{O}(n)$, then there exists an online algorithm $\mathcal{B}$ for the dynamic APSP problem with worst-case update/query time $n^{2/3-\Omega(\epsilon)}$. A precise statement and proof can be found in Section \ref{sec:reductionRealThing}.
\end{itemize}

\section{Related Work}

\paragraph{Dynamic Shortest Paths in Planar Graphs.} As stated in the introduction, the only two dynamic data structures for APSP in planar graphs are currently \cite{fakcharoenphol2001planar, klein2005multiple}. Both achieve update time $\tilde{O}(n^{\delta})$ and query time $\tilde{O}(n^{1-\delta/2})$ for $\delta \in [0, 1]$. In particular, for $\delta = \frac{2}{3}$, they achieve update and query time $\tilde{O}(n^{2/3})$. We note that various improvements followed that improved the above algorithms by logarithmic or doubly-logarithmic terms or generalized the result \cite{mozes2010shortest,kaplan2012submatrix, gawrychowski2018improved}.

A lower bound for the exact problem, given in \cite{abboud2016popular}, stipulates that 
$\tilde{\Omega}(n^{\delta})$ update time and 
$\tilde{\Omega}(n^{1-\delta})$ query time is necessary if one believes in the APSP conjecture. 

For undirected planar graphs in  the setting where one allows for a $(1+\epsilon)$-approximation on the distance queries, Abraham, Chechik and Gavoille \cite{abraham2012fully} presented a data structure that achieves worst-case update and query time $\tilde{O}(\sqrt{n} \log W/\epsilon^2)$ time where $W$ is the aspect ratio of $G$. Karczmarz \cite{karczmarz2018decrementai} achieved similar time bounds for the $(1+\epsilon)$-approximate APSP problem for a larger family of graphs (so-called $O(\sqrt{n})$-separable graphs), however, the algorithm only works for \emph{decremental graphs}. Further, the exact dynamic single-source shortest paths problem has been studied \cite{charalampopoulos2020single}. In very recent work by Filtser et al. \cite{filtser2024near}, a deterministic algorithm achieving subpolynomial amortized update and query time $n^{o(1)}/ \poly(\epsilon)$ was given.\\

\paragraph{Dynamic All-Pairs Shortest Paths in General Graphs.} The dynamic exact APSP problem has received extensive attention over the past years, also in incremental and decremental settings \cite{ausiello1991incremental, king1999fully, demetrescu2001fully, baswana2002improved, demetrescu2004new, thorup2005worst,  abraham2017fully,probstWulffNilsenwcAPSP, EvaldFGW21}. However, conditional lower bounds \cite{abboud2016popular, henzinger2015unifying} are quite strong and, in particular, stipulate that exact APSP requires either 
$\tilde{\Omega}(m)$ update or query time. This motivates the more restricted but powerful planar setting. \\

\paragraph{Dynamic Offline Algorithms.}
The purpose of studying dynamic algorithms in an offline setting is twofold. First it captures some of the inherent difficulties of the online setting, thus making any progress for the first provides significant insight towards designing an online dynamic algorithm. Secondly most of the lower bound results for standard dynamic setting also hold for offline setting. Thus, giving an efficient algorithm for the offline dynamic setting shows that no such conditional lower bound exists. Following this, substantial effort has been taken to design dynamic algorithms in offline setting. 
In ~\cite{CGHPS20} authors provide a fully dynamic offline algorithm that computes $O(\log^{4t} n)$ approximation of all pair max flow/min cut in average $\tilde{O}(m^{1/t +1})$ update and query time . They also propose a fully dynamic offline algorithm for all pair shortest path that computes $(2r-1)^t$ approximation in average $\tilde{O}(m^{1/t +1}n^{2/r})$ update and query time for $t,r\ge 1$. ~\cite{PSS17} presents a fully dynamic offline algorithm computing 3 edge and vertex connectivity in an undirected graph in time $O(\log n )$ per update. ~\cite{LS13,KL15} studies the dynamic connectivity and reachability in graph timelines in a semi offline setting. Here the updates are given upfront while the queries can arrive online. For the connectivity problem ~\cite{KL15} provides a data structure with $O(m+t\log n)$ preprocessing time and $O(\log n)$ query time. For rechability in undirected graph ~\cite{LS13} gives a randomized algorithm that has $O(t \log t \log \log t \log n+m)$ preprocessing time and can answer all queries in $O(\log n \log \log t)$ time.

\paragraph{Static Planar Graphs.} Finally, we point out that there has been extensive work on $(1+\epsilon)$-approximate and exact static distance oracles culminating recently in a data structure that can be constructed in $n^{1+o(1)}$ space and answer distance queries exactly in subpolynomial worst-case query time \cite{charalampopoulos2019almost, LongP21}.

\section{Preliminaries}

In this article, we assume that every graph $H$ is planar, directed and weighted with positive weights, is already embedded in the plane, is triangulated and undergoing edge weight updates given by an efficient encoding of the weight update sequence. We call a dynamic graph $H$, \emph{incremental}/ \emph{decremental} if the weight update sequence monotonically decreases/ increases edge weights in $H$ over time. In the offline setting, we assume that the update sequence is given with the input, in the online setting the sequence is revealed on-the-go. We use $G=(V,E,w)$ to denote the input graph, where $V$ is the set of vertices with $n = |V|$, $E$ is the set of edges and $w: E\rightarrow \mathbb{R}$ is the edge weight function. Assume constant time arithmetic on the weights. We let $d_G(u,v)$ denote the shortest path distance from $u$ to $v$ in $G$ and assume that all shortest paths at any time in $G$ are unique.
This can be ensured using the deterministic perturbation technique in~\cite{UniqueShortestPaths}).

Finally, we introduce the concept of full persistence. 

\begin{definition}[Fully-persistent]
We say a data structure undergoing an update sequence is \emph{fully-persistent} if every version of the data structure can be accessed and modified.
\end{definition}

\section{Online Incremental DDG}

In this section, we show the following result which we will use in our reductions in \cref{sec:reductionRealThing} and \cref{sec:reductionOffline} to get the main results of the paper.

\begin{theorem}\label{Thm:IncDDG}
Given input graph $G$, and an $r$-capped recursive decomposition $\mathcal{R}$ of $G$ with roots $G_1, G_2, \dots, G_k$.

Then, there exists a deterministic data structure that, given any $G_i$, maintains the DDGs of all pieces in the tree rooted at $G_i$. Its preprocessing time is $\tilde{O}(n)$ and it handles each edge weight decrease in $G_i$ in worst case update time $\tilde{O}(\sqrt{r})$. After every update, a query for the weight of any DDG edge of any piece in the tree rooted at $G_i$ can be answered in worst-case time $\tilde{O}(1)$. 

The data structure is fully-persistent.
\end{theorem}

In our construction, we focus on the first part of the theorem statement, and only at the end of this section we show how to modify our data structure to make it persistent.




\subsection{The Data Structure}

\paragraph{Preliminaries.} First, we need some additional notation. To avoid clutter, let $H = G_i$. We let the updates to $H$ be encoded by tuples $\{(e^1, \omega^1), (e^2, \omega^2), \dots \}$ where $e^i$ is an edge in $H$, and $\omega^i$ its new edge weight. 
In this section, we also need to argue about the embedding of $G$. We assume that we are given a combinatorial planar embedding, also often called rotation system, that specifies for each vertex $v$ in $G$, a cyclic order $e_0, e_1, \dots, e_{k-1}$ of the edges incident to $v$, such that in the planar embedding every $e_i$ is in between $e_{i-1}$ and $e_{i+1}$. We define a sharpest right-turn (left-turn) at $v$ from an edge $e$ to be the first out-going edge $e'$ incident to $v$ in counter-clockwise (clockwise) order starting in $e$. 

We further let for each piece $P$ (where we only talk for the rest of this section about pieces $P$ in the tree rooted at $H$), $P^t$ refer the version of the piece $P$ after the first $t$ updates have been applied to $H$. In particular $P^0$ is the initial piece.

For each piece $P$, we let $h(P) = \{h_1(P), h_2(P), \dots, h_{s_P}(P)\}$ be the constant number of holes such that each vertex in $\partial P$ is on some such hole. To ease indexing and avoid clutter, we allow $h_i(P)$ for any positive $i$, and let it refer to the $h_{i'}(P)$ such that $i'$ is the unique integer in $[1, s_h]$ derived from subtracting a multiple of $s_h$ from $i$. 

Further, for each hole $h \in h(P)$, we let $Q(h)$ denote a cyclic ordering $v_1(h), v_2(h),\ldots, v_{s_h}(h)$ of the (boundary) vertices of $h$ such that the interior of $h$ is to the right when traversing the vertices in this order. We also abuse notation slightly and denote by $[v_i(h), v_j(h)]$ for $1 \leq i \leq j \leq s_h$ the subset $\{v_i(h), v_{i+1}(h), \dots, v_j(h)\}$ of $Q(h)$ when the context is clear. Further, we also define a \emph{wrapped interval} $[v_i(h), v_j(h)]$ for the hole $h$, where if $1 \leq  i \leq j \leq s_h$ the interval is defined as above and if $1 \leq j < i \leq s_h$, is defined to be the interval $\{v_i(h), v_{i+1}(h), \dots, v_{s_h}(h), v_1(h), v_2(h), \dots, v_j(h)\}$. We define $(v_i, v_j)$ to be $[v_i, v_j] \setminus \{v_i, v_j\}$. We use the above indexing trick for $v_i$'s where $i$ is potentially larger than $s_h$ in conjunction with intervals, and wrapped intervals.

Finally, we require a simple data structure.

\begin{restatable}{theorem}{lastIntervalDet}[Last Interval Detection Data Structure]\label{thm:lastInterval}
Given a universe of size $U$, we can maintain a data structure that allows adding at time $i \geq 1$, a sub-interval $[l_i, r_i] \subseteq [1,U]$, and for queries at time $i$ where on input $x \in [1,U]$, the data structure outputs the largest integer $j \leq i$, such that $x \in [l_j, r_j]$ where we assume that $[l_0, r_0] = [1,U]$ (so $j$ is well-defined). If the length of the update sequence is polynomially bounded in $n$, then the data structure can be implemented with pre-processing time $\tilde{O}(1)$ and $\tilde{O}(1)$ worst-case update and query time.
\end{restatable}

\paragraph{Initialization.} We start by computing initially, for each piece $P$ in the tree of $\mathcal R$ rooted at $H$, the DDG using the FR-Dijkstra algorithm \cite{fakcharoenphol2001planar} (we later show that we need to store slightly more information to make updates efficient).  We further initialize for each piece $P$, boundary vertex $b \in \partial P$ and for each hole $h \in h(P)$ a Last Interval Detection data structure $\mathcal{D}_{b, h}$ on the universe $Q(h) = \{ v_1(h), v_2(h),\ldots, v_{s_h}(h)\}$ as described in Theorem \ref{thm:lastInterval}.\\

\paragraph{Updates.} After the $t^{th}$ update, we compute and store for each piece $P$, with $e^t = (u^t, v^t) \in P$, the distances $d_{P^t}(b,u^t)$ and $d_{P^{t}}(v^t, b)$ for each $b \in \partial P$. Further, for each such piece $P$, $b \in \partial P$, and hole $h \in h(P)$, we find the set of vertices $S_{b,h} \subseteq Q(h)$ such that for each $b' \in S_{b,h}$,
\[
d_{P^{t-1}}(b,b') > d_{P^t}(b,u^t) + \omega^t + d_{P^t}(v^t,b').
\]
We say that a vertex $b' \in S_{b,h}$ \emph{profits} from $e^t$.

We will later see that $S_{b,h}$ consists of at most two subintervals of $Q(h)$. We add these two subintervals to the data structure $\mathcal{D}_{b, h}$. Before we show how to implement updates efficiently, let us illustrate a query of the data structure.\\

\paragraph{Queries.} Consider a query for the weight of an edge in the DDG of some piece $P$ from a boundary vertex $b$ to a boundary vertex $b'$. 

We can now simply query the data structure $\mathcal{D}_{b, h}$ for vertex $b'$. This allows us to locate the last time $t$ that the shortest path from $b$ to $b'$ (in $P$) was decreased in weight (note that technically the times of $\mathcal{D}_{b, h}$ and versions of the pieces $P$ do not align however it is straight-forward to pair them). Since the decrease must have resulted from decreasing the weight of edge $e^t$, we can now simply output 
\[
d_{P^t}(b,b') = d_{P^t}(b,u^t) + \omega^t + d_{P^t}(v^t,b').
\]
Since we stored $d_{P^t}(b,u^t), d_{P^t}(v^t,b')$ and $\omega^t$ at time step $t$, a look-up suffices to recover them once $t$ is known. Since $b'$ has not profited from any edge thereafter, we conclude that the returned distances is the current distance.\\

\paragraph{Computing Distances for Updates Efficiently.} It remains to give the precise algorithm to update the data structure when processing the $t^{th}$ update to $H$. Here, we focus on updating a single affected piece $P$ (i.e. $e^t \in P$) in $\tilde{O}(\sqrt{|P|})$ time. It then follows immediately that we can update all affected pieces in time $\tilde{O}(\sqrt{r})$. 

We start by describing the first step of our update procedure, that is to compute for each vertex $x \in \partial P$, the distances $d_{P^t}(x,u^t)$ and $d_{P^t}(v^t, x)$. Note however, that we can find these distances in the graph $P^{t-1}$, since shortest paths to $u^t$ do not need to use an out-going edge of $u^t$. More precisely, these distances can be found in $P^t \setminus \{e^t\} = P^{t-1} \setminus \{e^t\}$. An analogous observation also holds for the distances $d_{P^t}(v^t, x)$, i.e. $d_{P^t}(v^t, x) = d_{P^{t-1}}(v^t, x)$. 

It thus remains to compute these distances in the graph $P^{t-1}$. To this end, let $P_1, P_2, \dots, P_k$ be the unique path in $\mathcal{R}$ from piece $P = P_1$ to the leaf piece $P_k$ containing $e^t$. Further, for each such $P_i$, let $P_{i,1}, P_{i,2}, \dots P_{i,k_i}$ be the siblings of $P_i$ in $\mathcal{R}$ (this set does not contain $P_i$ itself). Let $\overline{P_{i,1}}, \overline{P_{i,2}}, \dots, \overline{P_{i,k_i}}$ be the dense distance graphs of all siblings of $P_i$. 

Next, note that $\bigcup_{\ell,j} P^{t-1}_{\ell,j}$ forms exactly the graph $P^t \setminus \{e^t\}$ by our discussion above and the fact that each leaf of $\mathcal{R}$ contains exactly one edge (see Definition \ref{def:cappedRDecomp}). Further note that our data structure already holds the DDGs $\overline{P^{t-1}_{i,1}}, \overline{P^{t-1}_{i,2}}, \dots, \overline{P^{t-1}_{i,k_i}}$.

Then, to find the distances $d_{P^t}(v^t, x)$ for $x \in \partial P$, we can invoke the following version of FR-Dijkstra which is a stronger version of Theorem \ref{thm:FRDijkstra}.

\begin{theorem}[Extended FR-Dijkstra, see \cite{fakcharoenphol2001planar, kaplan2012submatrix}]\label{thm:FRDijkstraStronger}
Given a graph $G$, an $r$-capped recursive decomposition $\mathcal{R}$ a set of edge-disjoint pieces $X_1, X_2, \dots, X_k$ in $\mathcal{R}$, along with a data structure that answers queries for edge weights in their corresponding dense distance graphs $\overline{X_1}, \overline{X_2}, \dots \overline{X_k}$ in $\tilde{O}(1)$ time. 

Then, there exists an algorithm that given some $v \in \bigcup_{j} \partial X_{j}$, computes a shortest path tree $T_v$ in $\bigcup_{j} \overline{X_{j}}$ in time $\tilde{O}(\sum_{j} \sqrt{|X_j|})$. 
\end{theorem}

It is not hard to see that this not only gives all distance $d_{P^t}(v^t, x)$ when invoked for vertex $v^t$, but also with a single invocation of this algorithm, we can find all distances $d_{P^t}(x, u^t)$ by reversing all edges in the dense distance graphs and invoking the theorem again for $u^t$. Thus, the total time to find all such distances for our fixed piece $P$ is $\tilde{O}(\sqrt{|P|})$.\\

\paragraph{Updating $\mathcal{D}_{b,h}$ Efficiently.} Let us now describe how to update the data structures $\mathcal{D}_{b,h}$ for each $b \in \partial P$ and $h \in h(P)$. We first point out that we only need to update the data structures for affected pieces. 

Let us fix $P, b$ and $h$ for the rest of this section and let $h'$ be the hole in $h(P)$ with $b$ on $h'$. Note that for every $b' \in h$, we could query the distances $d_{P^{t-1}}(b,b')$ in the current version of the data structure. We could then verify using the information computed in the last paragraph whether $d_{P^{t-1}}(b,b') < d_{P^t}(b,u^t) + \omega^t + d_{P^t}(v^t, b')$ for each $b'$. However, we aim to identify the entire set $S_{b,h}$ of all such vertices $b' \in h$ in $\tilde{O}(1)$ time. To this end, we first need to prove that $S_{b,h}$ consists of at most two subintervals of $Q(h)$ so that we can represent this set by the $O(1)$ endpoints of such intervals. This will allow us to update each affected data structure $\mathcal{D}_{b,h}$ in $\tilde{O}(1)$ time and thus update all such affected data structures in time $\tilde{O}(\sqrt{r})$ since we only have to perform this operation once for each affected $P$ and every vertex $b$ on $h(P)$.

Here, we prove an even more powerful statement that we further build upon in the next section when we show how to compute $S_{b,h}$ efficiently. We need some additional notation and we start with the concept of a right-most shortest path (which we define in the standard way, see for example \cite{klein2005multiple}).

\begin{definition}[Right-most Shortest Path]\label{def:rightmostSP}
Given a piece $P$, a vertex $v$ in $P$, and distinct holes $h, h' \in h(P)$, let $T_{v, \mid h}$ be the shortest path tree from $v$ in $P$ restricted to having the set of leaf vertices consist only of the vertices in $h$.

Removing every edge and vertex of $P$ not belonging to $h$ or to $T_{v, \mid h}$ gives a subgraph with $|h|$ faces, excluding $h$ itself. For any such face $f$, the \emph{right-most (shortest) path} $\pi_{P,T_{v, \mid h},f} = \{ e_1, e_2, \dots , e_k\}$ in $T_{v, \mid h}$ with respect to $f$ is the path that is obtained by the following procedure:
\begin{itemize}
\item If $v$ is on $f$, let $e_0$ be an artificial edge $(x,v)$ with head $v$ and with tail $x$ being an artificial vertex in the interior of $f$. For $i = 1,2,\ldots,k$, recursively pick the $i^{th}$ edge $e_i$ on $\pi_{P,T_{v, \mid h},f}$ as the edge in $T_{v, \mid h}$ that makes the sharpest right turn at the head of $e_{i-1}$ with respect to $e_{i-1}$ itself. The recursion stops once a leaf vertex in $T_{v, \mid h}$ is reached (see the fat path in Figure~\ref{fig:r}(a)).

\item If $v$ is not on $f$, $e_1$ is the unique first edge on the path of $T_{v, \mid h}$ from $v$ to a vertex of $f$ and the remaining edges $e_2,e_3,\ldots,e_k$ of $\pi_{H,T_{v, \mid h},f}$ are obtained recursively as above.
\end{itemize}
\end{definition}

\begin{definition}\label{def:r}


Let $f$ be the face of the subgraph of $P$ in Definition~\ref{def:rightmostSP} such that $f$ contains $h'$. Then we define $r_{v,h,h'}$ to be the last vertex on the path $\pi_{P, T_{v, \mid h}, f}$; see Figure~\ref{fig:r}(a).
\end{definition}
\begin{figure}
\centerline{\scalebox{1.00}{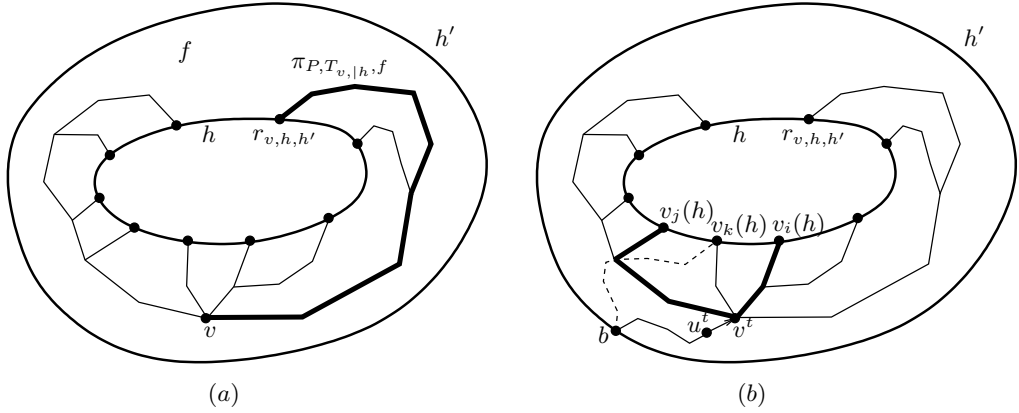}}
\caption{(a): Illustration of Definitions~\ref{def:rightmostSP} and~\ref{def:r}. In this example, the nine shortest paths from $v$ to $h$ induce nine faces (in addition to $h$) where $f$ is the face containing $h'$. The indices $i$ of vertices $v_i(h)$ are clockwise around $h$ in this drawing. (b): The contradiction reached in the proof of Lemma~\ref{lma:consecVerticesS}. The bold paths are $T_v[v_i(h)]$ and $T_v[v_j(h)]$ and the dashed path is a fictitious shortest path from $b$ to $v_k(h)$ that does not benefit from $(u^t,v^t) = e^t$.}\label{fig:r}
\end{figure}
Finally, let us prove the main statement of this section which implies that $S_{b,h}$ is the union of $O(1)$ intervals. This, together with an additional property below, allows us to apply binary search to find the endpoints of the intervals in $\tilde O(1)$ time, and then we can efficiently update $\mathcal D_{b,h}$.

\begin{lemma}\label{lma:consecVerticesS}
Given piece $P$, holes $h, h' \in h(P)$, and edge $e^t = (u^t, v^t) \in P$.

For any vertex $b$ on $h'$, and any two distinct vertices $v_i(h), v_j(h) \in S_{b,h} \subseteq Q(h)$ where $r_{v^t,h',h} \notin (v_i(h), v_j(h))$. Then $(v_i(h), v_j(h)) \subseteq S_{b,h}$.
\end{lemma}
\begin{proof}
We point out that we only consider the case where $b$ is on a hole $h' \neq h$. However, the case where $b$ on $h$ is simpler and can be straight-forwardly deduced from the description of the harder case by minor adaptions.
Assume for contradiction that a vertex $v_k(h) \in (v_i(h), v_j(h))$ is not in $S_{b,h}$, i.e. does not benefit from $e^t$. Then $d_{P^t}(b,u^t) + \omega^t + d_{P^t}(v^t,v_k(h))\leq d_{P^{t-1}}(b,v_k(h))$.

Let $T_{v^t}$ be the shortest path tree rooted at $v^t$ on $P$. By definition of $r_{{v^t},h',h}$ and by the assumption $r_{{v^t},h',h}\notin (v_i(h), v_j(h))$, every path from $b$ to $v_k(h)$ has to cross either the path $T_{v^t}[v_i(h)]$ or the path $T_{v^t}[v_j(h)]$; see Figure~\ref{fig:r}(b). In particular, any shortest path from $b$ to $v_k(h)$ that does not contain the edge $e^t$ has to cross one of these paths. By uniqueness of shortest paths, this is a contradiction.
\end{proof}

To see that $S_{b,h}$ consists of $O(1)$ intervals, it remains to observe that since $[v_i(h), v_j(h)]$ might be a wrapped interval, we might have to split it into two regular intervals before adding it to data structure $\mathcal{D}_{b,h}$.

\paragraph{Finding $S_{b,h}$ Efficiently.} Finally, we have to show how to compute $S_{b,h}$ efficiently, i.e. in $\tilde{O}(1)$ time. Similar to Lemma~\ref{lma:consecVerticesS}, we only describe how to find $S_{b,h}$ for $b$ on a hole $h' \neq h$. 
We now describe two essential sub-procedures to find $S_{b,h}$ efficiently:

First, we need to compute the vertex $r_{v^t,h',h}$ used in Lemma \ref{lma:consecVerticesS} in time $\tilde{O}(\sqrt{|P|})$ for each piece $P$ with $h, h' \in h(P)$. We give the proof in Appendix \ref{subsec:FRDijkstraLeftRight}.  

\begin{lemma}\label{lem:r_compute}
Given a piece $P$, distinct holes $h, h' \in h(P)$ and a vertex $v$ in $P$. We can find the vertex $r_{v,h,h'}$ in time  $\tilde{O}(\sqrt{|P|})$.
\end{lemma}
    
While $\tilde{O}(\sqrt{|P|})$ seems excessive, we only need to do such a computation for each affected piece $P$ and every ordered pair of holes in $h(P)$ once, so we can find all necessary vertices $r_{v^t,h',h}$ necessary in total time $\tilde{O}(\sqrt{|P|})$.
    
Secondly, we require the following tool.

\begin{restatable}{lemma}{tool}\label{lma:binarySearchTool}
Given any two vertices $v_i(h), v_j(h)$ on $h$, such that $v_i(h), v_j(h) \not\in S_{b,h}$, there is an algorithm that in $\tilde{O}(1)$ time detects whether the edge $e^t$ is contained in the region of the plane to the right of the directed cycle consisting of the shortest path $\pi_{b, v_i(h)}$ from $b$ to $v_i(h)$, the path in $h$ from $v_i(h)$ to $v_j(h)$ on $h$, and the reverse of the shortest path $\pi_{b, v_j(h)}$ from $b$ to $v_j(h)$ (see Figure~\ref{fig:r}(b)).
\end{restatable}
Before we prove Lemma \ref{lma:binarySearchTool}, let us show how we use these two sub-procedures to compute $S_{b,h}$ efficiently. Here, we can employ a simple binary search strategy. Let $r_{v^t,h',h}=v_k(h)$ in the cyclic ordering. Take two vertices $v_i(h), v_j(h) \in Q(h)$ where $v_i(h)$ is the vertex preceding $r_{v^t,h',h}$ in the cyclic ordering and $v_j(h)=v_{i-\lfloor \frac{s_h-(k-i)}{2} \rfloor}(h)$. Recall that we can test membership of $v_i(h), v_j(h)$ in $S_{b,h}$ in constant time, then 
\begin{itemize}
    \item if both are in $S_{b,h}$: we know that $[v_j(h), v_i(h)] \subseteq S_{b,h}$ and we continue the search in the interval $Q(h)\setminus [v_j(h), v_i(h)]$
    by setting new $v_j(h)=v_{j-\lfloor \frac{s_h-(k-j)}{2} \rfloor}(h)$.
    \item if neither of them are in $S_{b,h}$: use Lemma \ref{lma:binarySearchTool} to decide whether $e^t$ is in the region described in the Lemma. If so, we can conclude that $S_{b,h} \subseteq (v_j(h), v_i(h))$ and
    thus repeat the same process on $(v_j(h), v_i(h))$ by setting new  $v_j(h)=v_{i-\lfloor \frac{i-j}{2} \rfloor}(h)$ and new $v_i(h)=v_{i-1}(h)$; 
    otherwise, we know that $S_{b,h} \cap (v_i(h), v_j(h))=\emptyset$ and thus continue the search on $Q(h)\setminus (v_j(h), v_i(h))$. Here we set the new $v_i(h)$ to be the vertex following $r_{v,h',h}$ in the cyclic ordering and $v_j(h)=v_{j-\lfloor \frac{s_h-(k-j)}{2} \rfloor}(h)$.
    \item The case where exactly one vertex is in $S_{b,h}$ is similar. W.l.o.g assume $v_j(h)\in S_{b,h}$. We continue the search with new $v_i(h)=v_{i-\lfloor \frac{i-j}{2} \rfloor}(h)$ and new $v_j(h)=v_{j-\lfloor \frac{s_h-(k-j)}{2} \rfloor}(h)$.
    
\end{itemize}
Thus it follows that we can do binary search in $\tilde{O}(1)$ time using the above procedures.

\paragraph{Detecting $e^t$ in a Region.} It remains to prove the following Lemma.

\tool* 

\noindent\textit{Overview.} To prove the Lemma, we employ the following algorithm: first let us assume that the first edge on the shortest paths $\pi_{b, v_i(h)}, \pi_{b, v_j(h)}$ and the shortest path $\pi_{b,v^t}$ from $b$ to $v^t$ are all distinct. Here, we can use the planar embedding of $G$ to conclude whether $\pi_{b,v^t}$ uses as a first edge an edge to the right of $\pi_{b, v_i(h)}$ and to the left of $\pi_{b, v_j(h)}$ (or not) and if so we know that $e^t$ is in the region; otherwise $e^t$ must be outside the region.

Unfortunately, the paths $\pi_{b, v_i(h)}$ and $\pi_{b,v^t}$ might share the first few (or many) edges and $\pi_{b,v^t}$ only takes a right-turn after a while (the case where $\pi_{b, v_j(h)}$ and $\pi_{b,v^t}$ share many edges is analogous). We thus have to find the last vertex $x$ that is on both $\pi_{b, v_i(h)}$ and $\pi_{b,v^t}$. Deciding whether $\pi_{b,v^t}$ makes a left-turn after following $\pi_{b, v_j(h)}$ is then again simple using the embedding. Observe that the edge $e^t$ is not on $\pi_{b, v_i(h)}$ and $\pi_{b, v_j(h)}$ so that we can concentrate on finding the last common vertex on the shortest paths of $\pi_{b,u^t}$ and $\pi_{b, v_i(h)}$ and $\pi_{b, v_j(h)}$, respectively. \\

\noindent\textit{Storing Additional Information.} To find the vertex $x$, we need to store additional information. We use the following results.

\begin{theorem}[Level-Ancestor Data Structure, see \cite{bender2004level}]\label{thm:levelAnc}
Given a tree $T$ rooted at vertex $r$, there is a data structure that with preprocessing time $O(|T|)$ and query time $O(1)$ that on query for a vertex $x$ and a positive integer $d$ returns the $d^{th}$ edge on the $r$-to-$x$ path in $T$.
\end{theorem}

We assume that on initialization of our data structure, additionally to the DDG edges, we also store for each piece $P'$, and vertex $b' \in \partial P'$ the shortest path out-tree $T_{b'}$ over the lower level DDGs of $P'$, as can be computed by using Appendix \ref{sec:ExtendFRDijkstra}. Additionally, we know that $T_{b'}$ spans all boundary nodes of $P'$'s child pieces, denoted $\partial_{child} P'$. We mark these vertices in the tree and let $T'_{b'}$ be the tree derived from restricting the vertex set to $\partial_{child} P'$ (by recursively removing leaves that are not in $\partial_{child} P'$ and by contracting paths between two vertices in $\partial_{child} P'$). 

Then for each such tree, we initialize a level-ancestor data structure on $T'_{b'}$ as described in \Cref{thm:levelAnc}. We do the same for the in-tree $\overleftarrow{T}_{b'}$. It is not hard to observe that these computations can be done in time $O(|T_{b'}| + |\overleftarrow{T}_{b'}|) = \tilde{O}(\sqrt{|P|})$ time, thus the pre-processing can still be bound by $\tilde{O}(r)$.

After each update $(e^t, \omega^t)$, we do the same to the trees $T_{v^t}$ rooted at $v^t$ and the tree $\overleftarrow{T}_{u^t}$ rooted at $u^t$ computed on dense distance graphs with edge reversed. Again, this incurs at most $\tilde{O}(\sqrt{r})$ update time over all levels for each update. \\

\noindent\textit{Implementing the Query Efficiently.} Now, returning to the paths $\pi_{b, v_i(h)}$ and $\pi_{b,u^t}$, we distinguish two cases:
\begin{itemize}
    \item if no edge on $\pi_{b, v_i(h)}$ had its weight decreased since the data structure was initialized: then we use the trees $\overleftarrow{T}'_{v_i(h)}$ and $\overleftarrow{T}'_{u^t}$, starting at vertex $b$, we use binary search on the parameter $d$ in the level ancestor data structures to detect which boundary vertex $b'$ in $\partial_{child} P$ is the last on the shortest paths $\pi_{b, v_i(h)}$ and $\pi_{b,u^t}$ (from $b$ we use a reasonable $d$ for the paths to $v_i(h)$ and $u^t$ respectively and if the edges agree we search the interval containing values larger than $d$ and otherwise, we search the interval containing values smaller than $d$; we repeat this step iteratively to find the last edge that is shared on the paths by recursing on the relevant subintervals; this takes at most $O(\log n)$ time).
    
    Given $b'$, let $b''$ be the next vertex on $\overleftarrow{T}'_{v_i(h)}[b]$ and $b'''$ the next vertex on  $\overleftarrow{T}'_{u^t}[b]$. If $b''$ and $b'''$ are in different child pieces of $P$, then we can conclude that $b'$ is equal $x$. 
    
    Otherwise, we know that all three vertices are incident to the child piece $P''$. The final problem becomes to find the last common vertex of the shortest paths $\pi_{b', b''}$ from $b'$ to $b''$ and $\pi_{b', b'''}$ from $b'$ to $b'''$. 
    
    It is not hard to see that this problem can now be solved by recursing on the piece $P''$. Each binary search for the last child boundary vertex takes $O(\log n)$ time and the recursion depth is bounded by $O(\log n)$ giving a total query time of $O(\log^2 n)$.
    
    \item if some edge $e^{\ell}$ is the edge on $\pi_{b, v_i(h)}$ that had last its edge weight decreased: we can then use the same strategy as suggested above on the trees $T'_{v^{\ell}}$ and $\overleftarrow{T}'_{u^{\ell}}$ and the edge $e^{\ell}$ separately. While this increases the time to find the last child boundary vertex shared by the paths slightly, we still only require a single recursive call and therefore the query time is again $O(\log^2 n)$. 
\end{itemize}
This compeletes the description of our data structure and establishes Theorem \ref{Thm:IncDDG}.

\paragraph{Making the Data Structure Fully-Persistent.} Lastly, to turn our data structure into a fully-persistent data structure, we use the following result by \cite{driscoll1989making}.

\begin{definition}[Emphemeral Linked Data Structure]
An \emph{emphemeral linked data structure} is a data structure that consists of a finite number of \emph{nodes} where each node has a constant number of \emph{words} containing either information or a pointer to another node.  
\end{definition}
\begin{theorem}[\cite{driscoll1989making}]
Any emphemeral linked data structure, where each node has its pointer stored at a constant number of nodes at most, can be made fully-persistent at the cost of an additional \emph{factor} of $O(\log n)$ to the amortized/ worst-case update/ query cost.
\end{theorem}

While our described data structure is not an emphemeral linked data structure yet, we can convert multiple pointers pointing to the same node by binary trees over them which increases update times, query times and space by at most a logarithmic factor.

\paragraph{Acknowledgements.} The authors would like the anonymous reviewers of this article for their useful feedback and suggestions.



\bibliography{main}

\appendix

\section{Construction of the Last Interval Detection Data Structure}

Let us now show how the data structure presented in Theorem \ref{thm:lastInterval} can be constructed. We restate the theorem here for convenience.

\lastIntervalDet*

Assume that the length of the update sequence is bounded by $n^c$ for some (possibly large) constant $c$. Our algorithm maintains at all times an incomplete binary balanced tree $\mathcal{T}$ with $\lceil c \cdot \log_2(n) \rceil$ levels, where the leaf nodes are the intervals that where already added in the order that they arrived in. More precisely, after time $0$, $\mathcal{T}$ is a path of length $\lceil c \cdot \log_2(n) \rceil$ with the interval $[l_0, r_0] = [1,U]$ as its only leaf. Every time a new interval $[l_i, r_i]$ is added, we add it as a leaf to the right of interval $[l_{i-1}, r_{i-1}]$ into $\mathcal{T}$ possibly creating up to $O(\log(n))$ new internal nodes.

Now, additionally to maintaining the tree $\mathcal{T}$, we maintain at each internal node $v$ of $\mathcal{T}$, a dynamic interval tree (see for example \cite{cormen2009introduction}) which holds the intervals of the leaf nodes in the sub-tree rooted at the node. An interval tree has $O(\log n)$ update time and allows to query for any element $x \in [1, U]$ whether there is a segment in the interval tree that contains $x$ in time $O(\log n)$. An empty interval tree can be initialized in time $O(1)$.

It is not hard to see from this discussion, that on adding an interval $[l_i, r_i]$ to our tree, we can maintain the data structure by creating all nodes on the leaf-to-root path from $[l_i, r_i]$ to the root that are missing in $\mathcal{T}$, initializing a new interval tree at them, and finally adding $[l_i, r_i]$ to the interval tree of every node on the leaf-to-root path. Thus an update take worst-case update time $O(\log^2 n)$ time.

For a query with input $x$, we start from the root of $\mathcal{T}$, then query whether $x$ is in some interval in the subtree rooted at the right child of the root. If so, we move to the right child, otherwise, we move to the left child. We repeat this procedure on the node that we moved to and do so iteratively until we navigated to a leaf. It is not hard to see that the reached leaf node is exactly the last interval that contains $x$. Since we move along a path of at most $O(\log n)$ nodes and deciding whether to go to the left or right can be decided in $O(\log n)$ time using the interval trees, we have total query time $O(\log^2 n)$. 

This completes the proof.

\section{Extending FR-Dijkstra}\label{sec:ExtendFRDijkstra}
\subsection{Finding leftmost to rightmost shortest paths}\label{subsec:FRDijkstraLeftRight}
In this section, we describe an extension to FR-Dijkstra that allows it to obtain a vertex $r_{v,h,h'}$, given a vertex v in a piece with holes $h$ and $h'$. We will describe an extension that will not slow down FR-Dijkstra. Recall that $r_{v,h,h'}$ is found by growing a rightmost shortest path from $v$, taking the sharpest right turn at each step. After having executed FR-Dijksta to find a shortest path tree, the rightmost shortest path now is a sequence of DDG edges from the tree and we want to traverse these edges in order. At each step of the traversal, we need the cyclic ordering of outgoing tree edges that is consistent with the cyclic ordering in the underlying planar graph; this will allow us to pick the tree edge making the sharpest right turn. It thus suffices to give a preprocessing step which for each piece $P$ and each boundary vertex $u$ of $P$ computes and stores a cyclic ordering of DDG edges of $P$ outgoing from $u$. With this precomputed information, $r_{v,h,h'}$ can be found in time proportional to the shortest path tree found by FR-Dijkstra. Hence, the extension to FR-Dijkstra will not increase the overall time bound.

We may use FR-Dijkstra to precompute DDGs of all pieces bottom-up in the recursive decomposition where leaf pieces have constant size. Our preprocessing step processes pieces in the same bottom-up order. Let $P$ be a non-leaf piece $P$ and suppose that cyclic orderings have been found for all boundary vertices of all child pieces. Let $u$ be a boundary vertex of $P$. We will describe the preprocessing for $u$ that gives a cyclic ordering of DDG edges of $P$ outgoing from $u$. We let $h$ denote a hole of $P$ containing $u$.

Let $T$ be the shortest path tree from $u$ in $P$ found by FR-Dijkstra. This tree consists of DDG edges from child pieces of $P$. Given the precomputed cyclic orderings of these pieces, it is now fairly easy to find the ordering for $u$ in $P$. First, perform a DFS-traversal of $T$ where the next tree edge to be visited is the one that makes the sharpest right turn w.r.t.~the parent tree edge; in case the parent tree edge does not exist in $T$, we regard it as a dummy edge from a vertex embedded inside $h$ to $u$. Given the precomputed information for child pieces, the DFS-traversal of $T$ can be done in $O(|T|)$ time.  The cyclic ordering of DDG edges of $P$ from $u$ is now the same as the DFS ordering of the heads of these edges.

This completes the description of the preprocessing step and we argued for its correctness. The running time is dominated by the time to compute DDGs of all pieces using FR-Dijkstra. Hence, the preprocessing step will not increase the overall time bound.



\subsection{Dealing with distinct holes}\label{sec:FRDijkstraHoles}
In this subsection, we give some details missing in the paper by Fakcharoenphol and Rao on how to deal with distinct holes. These details are important in our setting since we are not treating FR-Dijkstra as a black box.

Let $h_1$ and $h_2$ be two \emph{distinct} holes of a piece $P$ and consider a setting where FR-Dijkstra is run on some union of DDGs one of which is $\DDG{P}$. The issue with relaxing edges from $h_1$ to $h_2$ is that the Monge property no longer holds since $h_1$ and $h_2$ are distinct faces. We show how to extend FR-Dijkstra to handle this.

We can view the vertices of $h_2$ as an interval $R$ defined by starting in some arbitrary vertex and then following the cyclic ordering of $h_2$ until all vertices have been visited. Let $L$ be the vertices of $h_1$ with the cyclic ordering from $h_1$. We wish to efficiently relax edges of the bipartite graph of DDG edges from $L$ to $R$.

Consider a step of FR-Dijkstra where a vertex $u\in L$ is activated. We need to argue how to efficiently find the vertices of $R$ that get $u$ as their new parent. Clearly, these vertices induce a subpath of $h_2$ and hence at most two sub-intervals of $R$. Using the same procedure as in the standard FR-Dijkstra procedure, we traverse $L$ in order from $u$ until a vertex $u'$ is found whose interval in $R$ is partially (but not fully) removed due to the activation of $u$.

FR-Dijkstra would now do binary search to determine which part of the interval of $u'$ should belong to the interval of $u$. However, this fails since we are missing the following property: if two vertices $v_1$ and $v_2$ in $R$ have $u$ as their new parent then every vertex between them in $R$ also has $u$ as a parent.

Let $r$ be the endpoint of the rightmost DDG edge of $u$ in $P$. 
Split $R$ into two sub-intervals that share the endpoint $r$. Then we can in fact ensure the property above by doing binary search in each of the two sub-intervals. This follows using similar arguments to those in the proof of Lemma~\ref{lma:consecVerticesS}. Since $r$ can be obtained in constant time using the auxiliary data computed in the extension to FR-Dijkstra in Section~\ref{subsec:FRDijkstraLeftRight}, we can thus deal with distinct holes without an asymptotic increase in the running time of FR-Dijkstra.

\section{Reduction from Dynamic DDG to Incremental DDG in the Offline Setting}
\label{sec:reductionOffline}
In this section, we give a simple reduction from the problem of maintaining a DDG on a dynamic graph to maintaining it on an incremental graph. For concreteness, we tailored the reduction to the guarantees of our incremental data structure but it is straight-forward to generalize the reduction further. From the reduction below, combined with Theorem \ref{Thm:IncDDG}, we immediately obtain our main Theorem \ref{thm:mainDDG} as a corollary.

\begin{theorem}
Given an offline data structure $\mathcal{I}$, that for any graph $H$, an $r$-capped recursive decomposition $\mathcal{R}_H$ of $H$ with roots $H_1, H_2, \dots, H_k$, and any $i$, maintains the DDGs of all nodes in the tree rooted at $H_i$, with preprocessing time $\tilde{O}(r)$, handles each edge weight decrease affecting $H_i$ in worst-case update time $\tilde{O}(\sqrt{r})$ and has query time $\tilde{O}(1)$.

Then there exists an \emph{offline} data structure $\mathcal{F}$, that given dynamic graph $G$ along with an $r$-capped recursive decomposition $\mathcal{R}$ with roots $G_1, G_2, \dots, G_{k'}$, for all $j$ maintains the DDGs of all pieces in the tree rooted at $G_j$ with pre-processing time $\tilde{O}(r)$ and amortized update time $\tilde{O}(\sqrt{r})$ and query time $\tilde{O}(1)$. 

If $\mathcal{I}$ is fully-persistent, we can turn the update and query time guarantees of $\mathcal{F}$ into worst-case guarantees.
\end{theorem}

We prove the theorem for the root nodes $G_1, G_2, \ldots, G_{k'}$. This is without loss of generality since the children of these nodes form an $r/2$-capped recursive decomposition and the data structure $\mathcal{I}$ can be applied to them and so on recursively, all at the cost of an additional $\tilde{O}(1)$ factor.

In the end of this section, we argue that standard de-amortization techniques can be used to achieve worst-case update and query time. Let us now state additional preliminaries and components required by our dynamic data structure.

\paragraph{Additional Notation.} We let $\tau$ be the number of updates that $G$ is undergoing and assume $\tau$ is a power of $2$ and be polynomially bounded in $n$. We denote by $G^t = (V, E, w^t)$ the graph $G$ \emph{after} the $t^{th}$ update is applied, and call each $G^t$ a \emph{version} of $G$. In particular, $G^0$ is the initial graph $G$. Further, we define $\lg x = \log_2 x$, $[x]$ to be shorthand for $\{0, 1, \dots, x-1\}$ and $\lfloor x \rfloor_i$ to be shorthand for $\lfloor x/2^i \rfloor \cdot 2^i$, i.e. $x$ rounded down to a power of $2^i$. 

\paragraph{A Graph Hierarchy.} We define a graph $G^t_i = (V,E,\overline{w}^t_i)$ for every $0 \leq i \leq \lg \tau$ and time $t$ to have the weight function defined for every $e \in E$ by 
$\overline{w}^t_i(e) = \max_{j \in [2^i]} w^{\lfloor t \rfloor_i + j}(e)$. We make the following straight-forward but nonetheless interesting observations.

\begin{observation}\label{obs:timeInvariant}
For every $0 \leq i \leq \lg \tau$ and time $t$ being a multiple of $2^i$, then for any $j \in [2^i]$, $G^t_{i+1} = G^{t+j}_{i+1}$. 
\end{observation}

\begin{observation}\label{obs:fewWeightIncreases}
For every $0 \leq i < \lg \tau$ and time $t$, we have that $G^{t}_i$ differs from $G^{t}_{i+1}$ by at most $2^i$ edge weight decreases. We denote these edge weight decreases by update sequence $\Delta_i^t$.
\end{observation}

For convenience, we also define $\Delta_i^t$ for negative $t$, to simply be the empty set.

\paragraph{A Top-Down Algorithm to Maintain the Graph Hierarchy Efficiently.} We have now defined all concepts necessary and are ready to state the algorithm. For initialization, at time $0$, we initialize an incremental data structure $\mathcal{I}$ on $G^0_{\lg \tau}$ and then apply in the following order the updates $\Delta^0_{\lg(\tau) - 1}, \Delta^0_{\lg(\tau) - 2}, \dots, \Delta^0_{0}$. The data structure $\mathcal{I}$ then holds the DDGs of all graphs $G_1, G_2, \ldots, G_{k'}$. 

Then, to process the next update to $G$, we then invoke $\textsc{AdvanceToNextVersion}()$. For queries on a particular version of $G^t$, we can after $t$ invocations of $\textsc{AdvanceToNextVersion}()$ simply query the data structure $\mathcal{D}$ that can answer DDG queries about the current graph.

\begin{algorithm2e}
$t \gets t+1$\;
Let $i_{max}$ be the largest integer such that $t \mod 2^{i_{max}} = 0$.\;
Ask $\mathcal{I}$ to roll back the weight update sequence $\Delta^{t-1}_0, \Delta^{t-1}_{1}, \dots, \Delta^{t-1}_{i_{max}}$.\label{lne:endforRollback}\;
Apply weight update sequence $\Delta^{t}_{i_{max}}, \Delta^{t}_{i_{max}-1}, \dots, \Delta^{t}_{0}$ to $\mathcal{I}$.\label{lne:reApply}
\caption{$\textsc{AdvanceToNextVersion}()$}\label{alg:updateVersion}
\end{algorithm2e}

The algorithm essentially exploits that for $j > i_{max}$, we have $H^{t-1}_{j} = H^{t}_{j}$ by the way we chose $i_{max}$ and by observation \ref{obs:timeInvariant}. This fact can be used to only roll back the updates $\Delta^{t-1}_0, \Delta^{t-1}_{1}, \dots, \Delta^{t-1}_{i_{max}}$ and after this step the data structure $\mathcal{D}$ holds effectively the graph $H^{t-1}_{i_{max}+1}$ which is exactly the graph $H^{t}_{i_{max}+1}$. At this point, it thus remains in the last update step in line \ref{lne:reApply}, to apply updates $\Delta^{t}_{i_{max}}, \Delta^{t}_{i_{max}-1}, \dots, \Delta^{t}_{0}$ such that at the end $\mathcal{D}$ holds the graph $H^t$ again as stipulated by observation \ref{obs:fewWeightIncreases}.

\paragraph{Analysis.} Correctness follows almost immediately from the discussion above, we merely point out that we have to unroll updates in the reverse order that we applied them to to ensure correctness, which is exactly what we do in our algorithm.

For the run-time analysis, we observe that the time to initialize $\mathcal{I}$ is $\tilde{O}(r + \tau \cdot \sqrt{r})$ by our previous discussion. Further, every $2^i$ updates, we roll-back $\Delta^{t-1}_i$ for some $t$ and add weight updates $\Delta^{t}_i$ to $\mathcal{I}$. We have by \Cref{obs:fewWeightIncreases} that each of these update sequences is of size at most $2^{i-1}$. Thus the data structure $\mathcal{I}$ can process these updates in time $\tilde{O}(2^i \cdot \sqrt{r})$. Thus, the amortized update time to handle such batches $\Delta^{t-1}_i$ for fixed $i$ at all times is $\tilde{O}(\sqrt{r})$. Note that our analysis was independent of the chosen $i$, and summing over all values of $i$, and adding the preprocessing gives amortized update time $\tilde{O}(\sqrt{r})$, as desired.

\paragraph{Improved Preprocessing.} It is not hard to see that we can partition the update sequence to $H$ into disjoint time segments of length $\lfloor \sqrt{r} \rfloor_2$. We then run the algorithm described above on each segment separately. It is easy to see that the above proof still gives the same guarantees while the preprocessing time of $\mathcal{I}$ for each segment is $\tilde{O}(r + \lfloor \sqrt{r} \rfloor_2 \cdot \sqrt{r}) = \tilde{O}(r)$.

\paragraph{De-Amortization.} To de-amortizate the data structure, we use that $\mathcal{I}$ is fully-persistent: every $t$ with $t \mod 2^i$ for some fixed $i$, we start building the data structure $\mathcal{I}$ for time $t+2^i$ from the version $t$ over the next $2^i$ updates where we spend $\tilde{O}(\sqrt{t})$ time at each time step. During this time, $\mathcal{I}$ is run further in its original version to support other updates. It is not hard to see that this increases the total update time by at most a logarithmic factor and ensures that the data structure achieves worst-case guarantees.

\section{Reduction from Dynamic DDG to Decremental DDG in the Online Setting}
\label{sec:reductionRealThing}

\begin{theorem}\label{thm:generalReduction}
Given an online data structure $\mathcal{D}$, that for any graph $H$, an $r$-capped recursive decomposition $\mathcal{R}_H$ of $H$ with roots $H_1, H_2, \dots, H_k$, and for any $i$, maintains the DDGs of all pieces in the tree rooted at $H_i$, with preprocessing time $T_{pre}(r)$, handles each edge weight increase affecting $H_i$ in amortized update time $T_u(r)$ and has query time $T_q(r)$. 

Then, there exists an \emph{online} data structure $\mathcal{F}$, that given a dynamic graph $G$, an $r$-division $G_1, G_2, \dots, G_{k'}$, maintains the DDG of each piece $G_{j}$ with pre-processing time $\tilde{O}(T_{pre}(r))$ and amortized update time $\tilde{O}(\sqrt{T_{pre}(r) \cdot T_q(r) \cdot \sqrt{r}} + T_u(r))$ and query time $\tilde{O}(T_q(r))$. 
\end{theorem}

The proof of the above theorem is rather straight-forward: for a fixed piece $P = G_j$, we initiate a decremental data structure $\mathcal{D}$ on $P$ and then rebuild every $\Delta$ updates to $P$, the entire data structure $\mathcal{D}$ on $P$. 

For convenience, at any time step, let us define $\mathcal{B}$ to be the batch of insertions that were applied to $P$ since the last time that $\mathcal{D}$ was rebuild. Then, at each time step, if an edge weight increase was issued by the adversary, the update is forwarded to $\mathcal{D}$. Independently of the type of update, we can then use the current DDGs of $\mathcal{D}$ and the batch $\mathcal{B}$ as inputs to the algorithm in Theorem \ref{Thm:IncDDG} to obtain a DDGs of the current graph $H$. 

It is easy to convince oneself of correctness. The running time consists of $T_{pre}(r)/ \Delta$ amortized time to rebuild the data structure $\mathcal{D}$, each update to $\mathcal{D}$ is processed in time $T_u(r)$ and finally, Theorem \ref{Thm:IncDDG} with the guarantee that $\mathcal{B}$ never exceeds size $\Delta$ (since we rebuild the data structure and reset $\mathcal{B}$ after $\Delta$ updates) ensures that the final DDG data structure that can be queried is created in time $\tilde{O}(\Delta \cdot \sqrt{r} T_q(r))$ and query time $\tilde{O}(T_q(r))$. By choice of $\Delta = \sqrt{\frac{T_{pre}(r)}{T_q(r) \cdot \sqrt{r}}}$, we obtain the stated amortized update time.

\section{Dynamic max flow}\label{sec:DynMaxFlow}
In this section, we prove Corollary~\ref{cor:MaxFlow}. We assume that the reader is familiar with the details of~\cite{ItalianoNSW11}. It is useful to separate the running time of their data structure into two parts: $\tilde O(n/\sqrt r + r)$ worst-case and $\tilde O(n/\sqrt r)$ amortized time. We focus on each part separately.

The $\tilde O(n/\sqrt r)$ amortized bound comes from the data structure having to recompute the $r$-division every $\sqrt r$ updates. This recomputation is necessary since~\cite{ItalianoNSW11} supports edge insertions and deletions and an inserted edge might have endpoints in distinct pieces, violating the properties of an $r$-division. The way this is handled is to regard each such edge as a trivial piece consisting of only that edge. The issue then is that the number of pieces and the number of boundary vertices of a piece may grow over time. However, by recomputing every $\sqrt r$ updates, the properties of the $r$-division are maintained. Since we only allow edge weight changes, this issue cannot arise so we can immediately eliminate the amortized time bound.

The $\tilde O(n/\sqrt r + r)$ worst-case bound is the time it takes to
\begin{enumerate}
    \item recompute the DDG of the piece containing the edge whose weight was changed,
    \item recompute candidate min $st$-separating cycles fully contained in that piece, and
    \item running the coarse version of Reif on the set of boundary vertices of the $r$-division.
\end{enumerate}
The first part is immediately improved from $\tilde O(r)$ to $\tilde O(\sqrt r)$ using our new data structure. The second part can be solved by having a recursive max flow/min $st$-cut data structure for each piece. The final part remains unchanged and runs in time $\tilde O(n/\sqrt r)$ using FR-Dijkstra. We thus get the following recurrence for the improved worst-case time: $T(n) = \tilde O(n/\sqrt r + \sqrt r) + T(r)$. Picking $r$ to be a constant factor smaller than $n$, we get the desired $\tilde O(\sqrt n)$ worst-case bound from a geometric sums argument.
\end{document}

%% file: defs.tex
\usepackage[usenames,dvipsnames,svgnames,table]{xcolor}
\usepackage{enumitem}
\usepackage{graphicx}
\usepackage{fancybox}
\usepackage{comment}
\usepackage[disable]{todonotes}
\usepackage{xcolor}
\usepackage{nameref}
\usepackage{epsfig}
\usepackage{lscape}
\usepackage{lineno}

\definecolor{ForestGreen}{rgb}{0.1333,0.5451,0.1333}
\definecolor{DarkRed}{rgb}{0.8,0,0}
\definecolor{Red}{rgb}{1,0,0}
\usepackage[linktocpage=true,
pagebackref=true,colorlinks,
linkcolor=ForestGreen,citecolor=ForestGreen,
bookmarks,bookmarksopen,bookmarksnumbered]
{hyperref}

\usepackage{amsmath}
\usepackage{amssymb}
\usepackage{amsthm}

\usepackage[ruled, linesnumbered, algo2e]{algorithm2e}
\usepackage[noend]{algorithmic}
\usepackage{subfig}

\usepackage{thm-restate}

\usepackage{float}
\usepackage{cleveref}

\newtheorem{theorem}{Theorem}[section]

\newtheorem{corollary}[theorem]{Corollary}
\newtheorem{lemma}[theorem]{Lemma}
\newtheorem{observation}[theorem]{Observation}

\newtheorem{definition}[theorem]{Definition}
\newtheorem{remark}[theorem]{Remark}

\newtheorem*{theorem*}{Theorem}
\newtheorem*{corollary*}{Corollary}
\newtheorem*{conjecture*}{Conjecture}
\newtheorem*{lemma*}{Lemma}
\newtheorem*{thm*}{Theorem}
\newtheorem*{prop*}{Proposition}
\newtheorem*{obs*}{Observation}
\newtheorem*{definition*}{Definition}

\newtheorem*{remark*}{Remark}
\newtheorem*{rec*}{Recommendation}

\newenvironment{fminipage}%
  {\begin{Sbox}\begin{minipage}}%
  {\end{minipage}\end{Sbox}\fbox{\TheSbox}}



%% file: r.pdf_t
\begin{picture}(0,0)%
\includegraphics{r.pdf}%
\end{picture}%
\setlength{\unitlength}{4144sp}%
\begingroup\makeatletter\ifx\SetFigFont\undefined%
\gdef\SetFigFont#1#2#3#4#5{%
  \reset@font\fontsize{#1}{#2pt}%
  \fontfamily{#3}\fontseries{#4}\fontshape{#5}%
  \selectfont}%
\fi\endgroup%
\begin{picture}(6044,2452)(5044,-16640)
\put(6751,-14596){\makebox(0,0)[lb]{\smash{{\SetFigFont{9}{10.8}{\familydefault}{\mddefault}{\updefault}{\color[rgb]{0,0,0}$\pi_{P, T_{v,\mid h},f}$}%
}}}}
\put(9381,-16196){\makebox(0,0)[lb]{\smash{{\SetFigFont{9}{10.8}{\familydefault}{\mddefault}{\updefault}{\color[rgb]{0,0,0}$v^ t$}%
}}}}
\put(9672,-15026){\makebox(0,0)[lb]{\smash{{\SetFigFont{9}{10.8}{\familydefault}{\mddefault}{\updefault}{\color[rgb]{0,0,0}$r_{v,h,h'}$}%
}}}}
\put(10756,-14461){\makebox(0,0)[lb]{\smash{{\SetFigFont{9}{10.8}{\familydefault}{\mddefault}{\updefault}{\color[rgb]{0,0,0}$h'$}%
}}}}
\put(6231,-16196){\makebox(0,0)[lb]{\smash{{\SetFigFont{9}{10.8}{\familydefault}{\mddefault}{\updefault}{\color[rgb]{0,0,0}$v$}%
}}}}
\put(6522,-15026){\makebox(0,0)[lb]{\smash{{\SetFigFont{9}{10.8}{\familydefault}{\mddefault}{\updefault}{\color[rgb]{0,0,0}$r_{v,h,h'}$}%
}}}}
\put(7606,-14461){\makebox(0,0)[lb]{\smash{{\SetFigFont{9}{10.8}{\familydefault}{\mddefault}{\updefault}{\color[rgb]{0,0,0}$h'$}%
}}}}
\put(9406,-16576){\makebox(0,0)[lb]{\smash{{\SetFigFont{9}{10.8}{\familydefault}{\mddefault}{\updefault}{\color[rgb]{0,0,0}$(b)$}%
}}}}
\put(6256,-16576){\makebox(0,0)[lb]{\smash{{\SetFigFont{9}{10.8}{\familydefault}{\mddefault}{\updefault}{\color[rgb]{0,0,0}$(a)$}%
}}}}
\put(8578,-16232){\makebox(0,0)[lb]{\smash{{\SetFigFont{9}{10.8}{\familydefault}{\mddefault}{\updefault}{\color[rgb]{0,0,0}$b$}%
}}}}
\put(9249,-15571){\makebox(0,0)[lb]{\smash{{\SetFigFont{9}{10.8}{\familydefault}{\mddefault}{\updefault}{\color[rgb]{0,0,0}$v_k(h)$}%
}}}}
\put(8944,-15486){\makebox(0,0)[lb]{\smash{{\SetFigFont{9}{10.8}{\familydefault}{\mddefault}{\updefault}{\color[rgb]{0,0,0}$v_j(h)$}%
}}}}
\put(9109,-16166){\makebox(0,0)[lb]{\smash{{\SetFigFont{9}{10.8}{\familydefault}{\mddefault}{\updefault}{\color[rgb]{0,0,0}$u^t$}%
}}}}
\put(9379,-15034){\makebox(0,0)[lb]{\smash{{\SetFigFont{9}{10.8}{\familydefault}{\mddefault}{\updefault}{\color[rgb]{0,0,0}$h$}%
}}}}
\put(6219,-15036){\makebox(0,0)[lb]{\smash{{\SetFigFont{9}{10.8}{\familydefault}{\mddefault}{\updefault}{\color[rgb]{0,0,0}$h$}%
}}}}
\put(9615,-15564){\makebox(0,0)[lb]{\smash{{\SetFigFont{9}{10.8}{\familydefault}{\mddefault}{\updefault}{\color[rgb]{0,0,0}$v_i(h)$}%
}}}}
\put(6076,-14534){\makebox(0,0)[lb]{\smash{{\SetFigFont{9}{10.8}{\familydefault}{\mddefault}{\updefault}{\color[rgb]{0,0,0}$f$}%
}}}}
\end{picture}%